\documentclass[a4paper,10pt]{amsart}
\usepackage[utf8]{inputenc}
\usepackage{amsmath,amsfonts,amssymb,amsthm,amscd,amsbsy,epsfig,array,mathrsfs,tikz,tikz-cd}
\usetikzlibrary{positioning}
\usetikzlibrary{decorations.pathmorphing}
\usetikzlibrary{decorations.markings}

\usepackage[perpage]{footmisc}

\newcommand{\cupdot}{\mathbin{\dot{\cup}}}

\newcommand{\m}[1]{\mathcal{#1}}
\newcommand{\mb}[1]{\mathbb {#1}}

\newcommand{\ti}[1]{\widetilde {#1}}

\newcommand{\ol}[1]{\overline{#1}}

\newcommand{\rest}[1]{_{\big| {#1}}}

\theoremstyle{plain}
\newtheorem{thm}{Theorem}[section]
\newtheorem*{thm*}{\bf Theorem }
\newtheorem{lem}[thm]{Lemma}
\newtheorem{prop}[thm]{Proposition}
\newtheorem*{prop*}{\bf Proposition}
\newtheorem{cor}[thm]{Corollary}

\theoremstyle{definition}
\newtheorem{defn}[thm]{Definition}

\newtheorem{example}[thm]{Example}

\theoremstyle{remark}
\newtheorem{rem}[thm]{Remark}

\numberwithin{equation}{section}

\title{Feynman amplitudes on moduli spaces of graphs}
\author{Marko Berghoff}

\begin{document}

\begin{abstract}
This article introduces moduli spaces of coloured graphs on which Feynman amplitudes can be viewed as ''discrete'' volume densities. The basic idea behind this construction is that these moduli spaces decompose into disjoint unions of open cells on which parametric Feynman integrals are defined in a natural way.  
Renormalisation of an amplitude translates then into the task of assigning to every cell a finite volume such that boundary relations between neighboring cells are respected. It is shown that this can be organized systematically using a type of Borel-Serre compactification of these moduli spaces. The key point is that in each compactified cell the newly added boundary components have a combinatorial description that resembles the forest structure of subdivergences of the corresponding Feynman diagram.  
\end{abstract}
\maketitle

\section{Introduction}

Understanding the analytic structure of functions defined by Feynman integrals is a long standing open problem in quantum field theory. Although many techniques and folklore theorems are being used in everyday practical calculations, our theoretical understanding of these structures is still far from satisfying. For instance, Cutkosky's theorem on branch cuts and monodromies of Feynman integrals \cite{cr} has been used in calculations for decades, but was proven only recently with the help of algebro-geometric methods in \cite{bk-cr}. In the process, Bloch and Kreimer mention a new idea to approach further studies of analytic structures in Feynman integrals using \textit{Outer space} (and related spaces), a construction from geometric group theory \cite{cv}. 

Inspired by Teichm\"uller theory, the basic idea behind Outer space $CV_n$ and its variants is to study automorphisms of free groups $F_n$ by their action on geometric objects, in this case built out of combinatorial graphs of rank $n$ equipped with additional (topological) data. These spaces and the corresponding actions have nice properties, adding geometric and topological methods to the group theorist's toolbox. One such property is that the action projects onto an action of $Out(F_n)$, the group of outer automorphisms of $F_n$, which acts on $CV_n$ with finite stabilizers. Since Outer space is contractible, it follows that the orbit space of this action, the moduli space of rank $n$ metric graphs, is a rational classifying space for $Out(F_n)$. It encodes thus its rational homology.

In \cite{hv} the homology of $Aut(F_n)$ is computed utilizing a cubical cell structure of the corresponding moduli space of rank $n$ graphs with a marked basepoint (in this case inner automorphisms act non-trivially). Quite surprisingly, the results in \cite{bk-cr} show that the same structure is found in the study of poles and branch cuts of Feynman integrals; the combinatorial operations involved in determining these critical subsets in the space of external momenta of a given Feynman diagram $G$, contracting subsets of its edge-set and putting edge-propagators in the Feynman-integrand on mass-shell, form a similar chain complex of cubes.
\newline

The aim of this article is to add another observation to the list of connections between these two so-far unrelated\footnote{A relation between the underlying combinatorial structures of the constructions in \cite{bv} and \cite{bek} was already noted in \cite{bk}, but not further pursued.} fields; the similarity between certain bordifications of spaces of graphs, as in \cite{bv,v-bord}, and the algebraic geometer's approach to renormalisation of Feynman integrals, as in \cite{bk-as}, based on the methods of \cite{bek}. 

The basic idea is that each Feynman integral $I_G$ can be interpreted as the volume of a cell $\sigma_G$ in an appropriate moduli space of graphs. If the integral is divergent, all its divergences sit on certain faces of $\sigma_G$ or, in the language of moduli spaces, at infinity. Thus, renormalisation translates in this formulation into the task of rendering this integral convergent at infinity. This can be formulated conveniently using distributions on $\sigma_G$. First, the cell $\sigma_G$ is compactified in the sense of Borel-Serre, in order to have better control of the behaviour of $I_G$ at infinity, then the necessary subtractions are employed to take care of the divergences, now situated at the boundary of the compactified cell, in accordance with the usual renormalisation schemes.

Moreover, the nature of these moduli spaces of graphs allows to treat all integrals corresponding to a given rank and number of external edges at once, so that we can formulate Feynman amplitudes - albeit a rather unphysical version - as generalized distributions on these spaces. 
Roughly speaking, one sums over each cell $\sigma_G$, where $G$ is a graph of rank $n$ with $k$ external edges labeled by an external momentum configuration $p$, integrated against a density $\omega_G$ (depending on $p$) that is determined from $G$ by Feynman rules in their parametric representation,
\begin{equation*}
\textit{(unrenormalised) $n$-loop contribution to } \m A(p) = \sum_{\mathit{rank}(G)= n} \langle \omega_G(p) \mid \sigma_G \rangle.
\end{equation*}

To formulate this precisely and extend it to a renormalised version is the goal of the present article. The essential ingredient for this to work is the equivalence of the combinatorics behind renormalisation and the above mentioned compactification method. 
\newline

The article is organized as follows. In Section \ref{prelim} we set up some basic notation that will be needed throughout the text. The following two sections serve as a very!\ short introduction to the central topics, Feynman integrals and renormalisation on one side and moduli spaces of graphs and their bordifications on the other side. Since the focus lies on the combinatorial aspects behind both constructions, the exposition is kept rather basic; for technical details or a more thorough introduction on each individual topic the interested reader is invited to consult the given literature references. Section \ref{pscpwd} introduces the notions of piecewise distributions and pseudo complexes which allow to define a sort of discrete integration theory on such spaces. The next section connects all the previously introduced concepts by applying this theory to the case of Feynman integrals in their parametric formulation and moduli spaces of metric (coloured) graphs. The final section finishes with a discussion of the renormalisation problem and its solution.

\section{Preliminaries}\label{prelim}
We start with some basic definitions and notation conventions.

\begin{defn}\label{graphs}
A graph $G$ is a quadruple $G=(V,H,s,c)$ where $V$ is the set of vertices of $G$ and $H$ its set of half-edges. The map $s:H\rightarrow V$ attaches each half-edge to its source vertex, the map $c:H\rightarrow H$ connects half-edges and satisfies $c^2=id_H$. If $c(h)=h'\neq h$ the pair $e=\{h,h'\}$ is called an internal edge of $G$. We denote the set of internal edges of $G$ by $E=E(G)$ and its cardinality by $N=N_G$. The remaining half-edges, satisfying $c(h)=h$, are called external edges or legs or hairs.

An (internal) edge subgraph $\gamma \subset G$ is determined by a subset $E(\gamma)\subset E(G)$ of the internal edges of $G$. The vertex set of $\gamma$ consists of all vertices of $G$ that are connected to edges of $\gamma$. So $\gamma$ is a graph itself without external edges.
\end{defn}
\begin{rem}
In the following it will be convenient to retreat to the ''usual'' definition of combinatorial graphs, i.e.\ as tuples $(V,E)$ with an attaching map $\partial: E \rightarrow Sym^2(V)$ and treat legs merely as additional data. In Section \ref{osmsog} where we take a topological point of view we think of graphs simply as of one dimensional CW-complexes. In this case legs can be modeled either by introducing auxiliary external vertices of valence one or as additional labels on the vertex set $V$.
\end{rem}
We need two operations on graphs throughout this work, the contraction and deletion of subgraphs.
\begin{defn}
Let $G$ be a graph and $\gamma \subset G$ a connected subgraph. The contracted graph $G/\gamma$ is given by replacing $\gamma$ by a vertex and connecting each edge in $E(G)\setminus E(\gamma)$ with it. If $\gamma$ is a disjoint union of subgraphs the contraction is defined componentwise.

The deletion of $\gamma$ in $G$ is the graph $G \setminus \gamma$ with $V(G \setminus \gamma)= V(G)$ but all edges in $E(\gamma)$ removed, $E(G \setminus \gamma)= E(G)\setminus E(\gamma)$.
\end{defn}

Some special types of graphs:
\begin{defn} Let $G$ be a graph. Its rank or loop number will be denoted by $|G|:=h_1(G)=|H_1(G)|.$
\begin{enumerate}
 \item[1.] $G$ is called core or $1$PI if removing any edge reduces its rank, $|G\setminus e|<|G|$.
 \item[2.] A forest in $G$ is a subgraph $T \subset G$ with $|T|=0$. If $T$ is connected it is called a tree.
 \item[3.] A forest or tree in $G$ is spanning if its vertex sets equals $V=V(G)$.
 \item[4.] A rose graph with $n$ petals is a graph $R_n$ with one vertex and $n$ internal edges. The case $n=1$ is known as a tadpole in physics.
 \item[5.] An (proper) edge-colouring of $G$ is a map $c:E(G) \rightarrow C$ that assigns to every edge $e\in E(G)$ a colour $c(e)$ in a set of colours $C$ such that no two adjacent edges are assigned the same colour.
\end{enumerate}
\end{defn}

\section{Feynman integrals}
\subsection{Parametric Feynman integrals}\label{feyint}
Let $G$ be a connected graph with $N$ internal and $k$ external edges. We refer to $G$ as a Feynman diagram if it is equipped with additional physical data. It describes then a term in the perturbative expansion of some physical quantity, typically a particle scattering process. Here we consider the case where one associates to every internal edge a mass $m_e\geq 0$ and to each leg a momentum $p_i \in \mb R^d$. The $p_i$ are vectors in $d$-dimensional Minkowski (or Euclidean) spacetime for even $d\in 2\mb N$ and satisfy \textit{momentum conservation}
\begin{equation*}
\sum_{i=1}^k p_i = 0.
\end{equation*}
We abbreviate this external data by $(p,m)$. 
\newline

\textit{Feynman rules} assign to a graph $G$, labeled by $(p,m)$, the integral
\begin{equation}\label{fi}
I_G(p,m) := \int_{\sigma_G} \omega_G(p,m),
\end{equation}
where 
\begin{equation*}
 \sigma_G=\mb P (\mb R_+^N)= \{[x_1: \ldots : x_N] \mid x_i \geq 0 \}\subset \mb P(\mb R^N)
\end{equation*}
is the subset of projective space formed by all points with non-negative homogeneous coordinates and the differential form $\omega_G$ is defined using two graph polynomials as follows.
\begin{defn}\label{graphpolys}
Let $G$ be a connected graph. The first Symanzik (or Kirchhoff) polynomial is defined as
\begin{equation*}
 \psi_G= \sum_{T \subset G} \prod_{e \notin  T} x_e,
\end{equation*}
where the sum is over all spanning trees of $G$.

The second Symanzik polynomial is defined as 
\begin{equation*}
 \phi_G= \sum_{T=T_1\cup T_2 \subset G} (p^{T_1})^2 \prod_{e \notin T_1 \cup T_2}x_e,
\end{equation*}
where the sum is now over all spanning $2$-forests $T=T_1\cup T_2$ - a spanning $2$-forest is a disjoint union of two trees $T_1$ and $T_2$ in $G$ with $V(G)=V(T_1)\cupdot V(T_2)$ - and
\begin{equation*}
 p^{T_1}:= \sum_{v \in V(T_1)} p_v
\end{equation*}
is the sum of all external momenta entering the component of $G$ that is spanned by $T_1$. By momentum conservation, it equals $-p^{T_2}$.

If $G=G_1 \cupdot \ldots \cupdot G_k $ is a disjoint union of graphs, then $\psi_G$ and $\phi_G$ are defined by
\begin{equation*}
 \psi_G= \prod_{i=1}^k \psi_{G_i} , \quad  \phi_G= \sum_{i=1}^k  \phi_{G_i} \prod_{j \neq i}^k\psi_{G_j} . 
\end{equation*}
\end{defn}

For more on these polynomials and how renormalisability of Feynman integrals crucially depends on some of their properties, see \cite{bk-as}. We cite two important relations in 
\begin{prop}\label{condel} Let $G$ be connected. Then
 \begin{equation}\label{edgeco}
 {\psi_G}\rest{x_e=0} = \psi_{G / e}  ,\quad  {\phi_G}\rest{x_e=0}= \phi_{G/e},
 \end{equation}
 and
 \begin{equation}\label{scaling}
        \psi_G=\psi_{\gamma} \psi_{G / \gamma} + R_\gamma,\quad \phi_G= \psi_\gamma \phi_{G / \gamma} + R_\gamma ',
       \end{equation}
       where $R_\gamma$ and $R'_\gamma$ are both of degree strictly greater than $deg(\psi_\gamma)=|\gamma|$ in the variables $x_e$, $e \in E(\gamma)$.
\end{prop}
\begin{proof}
 Both statements follow from Definition \ref{graphpolys} by partitioning the set of all spanning trees or $2$-forests of $G$ into those that do or do not intersect with $\gamma$.  
\end{proof}

Finally, let $\Xi_G$ denote the polynomial
\begin{equation*}
 \Xi_G= \phi_G + \psi_G \sum_{i=1}^N m_i^2 x_i 
\end{equation*}
and define the differential form $\omega_G$ by
\begin{equation}\label{omega}
 \omega_G(p,m)= \psi_G^{-\frac{d}{2}}\left( \frac{\psi_G}{\Xi_G(p,m)} \right)^{N-|G|\frac{d}{2}} \nu_G=:f_G(p, m)\nu_G
\end{equation}
with 
\begin{equation*}
 \nu_G=\nu_N = \sum_{i=1}^N (-1)^i x_i dx_1 \wedge \ldots \wedge \widehat{dx_i} \wedge \ldots \wedge dx_N.
\end{equation*}

\begin{example}
Let $G$ be the ''Dunce's cap'' graph, depicted in Figure \ref{dunce}. In $d=4$ we have $N-|G|\frac{d}{2}=0$ and $ \psi_G=x_3x_4 + x_2x_3 + x_2x_4 + x_1x_3 + x_1x_4$, so that 
\begin{equation*}
 I_G(p,m)=\int_{\mb P(\mb R_+^4)} \frac{\nu_4(x_1,x_2,x_3,x_4)}{(x_3x_4 + x_2x_3 + x_2x_4 + x_1x_3 + x_1x_4)^{2}}.
\end{equation*}
Note that the denominator vanishes for $x_3=x_4=0$ rendering the integral divergent. This is a general phenomenon which we discuss in the next section.
   \begin{figure}[!htb]
  \begin{tikzpicture}
   \coordinate (v0) at (0,0);
   \fill[black] (v0) circle (.066cm);
   \coordinate  (v1) at (2,-1);
   \fill[black] (v1) circle (.066cm);
   \coordinate (v2) at (2,1);
   \fill[black] (v2) circle (.066cm); 
   \coordinate (p3) at (2.7,-1.35);
   \coordinate (p4) at (2.7,1.35);
   \coordinate (p1) at (-0.7,0.35);
   \coordinate (p2) at (-0.7,-0.35);
   \draw (v1) -- (p3) node [xshift=0.3cm,yshift=-0.1cm] {$p_3$};
   \draw (v2) -- (p4) node [xshift=0.3cm,yshift=-0.1cm] {$p_4$};
   \draw (v0) -- (p1) node [xshift=-0.2cm,yshift=0.1cm] {$p_1$};
   \draw (v0) -- (p2) node [xshift=-0.2cm,yshift=-0.2cm] {$p_2$};
   \draw (v0) -- (v1) node [midway,below] {$m_2$};
   \draw (v0) -- (v2) node [midway,above] {$m_1$};
   \draw (v2) to[out=-135,in=135] (v1) node [left,xshift=0.22cm,yshift=0.9cm] {$m_3$};
   \draw (v2) to[out=-45,in=45] (v1) node [right,xshift=0.33cm,yshift=0.9cm] {$m_4$};  
  \end{tikzpicture}
\caption{Dunce's cap}\label{dunce}
 \end{figure}
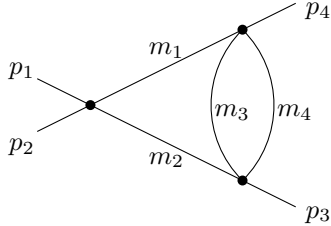
\end{example}

\subsection{Renormalisation}\label{renorm}

In general the integral $I_G$ in \eqref{fi} is ill-defined; $f_G$ may have non-integrable singularities at the loci where certain subsets of edge variables vanish\footnote{A possible overall divergence is in the projective representation we are using here hidden in a prefactor of $I_G$, cf.\ Remark \ref{rem6}.}. The condition for such \textit{ultraviolet} divergences to appear can be phrased in terms of the subgraph that is spanned by the edges corresponding to these variables. It depends only on the topology of that subgraph through its \textit{superficial degree of divergence}
\begin{equation}\label{superficial}
s_\gamma= d|\gamma|-2N_\gamma.
\end{equation}

There is also the possibility of so-called \textit{infrared} divergences which we avoid here by considering only massive diagrams (all $m_i>0$) or \textit{generic external momentum configurations},
\begin{equation}\label{gem} 
\big(\sum_{i \in I}p_i \big)^2>0 \text{ for all proper subsets } \emptyset \neq I \subsetneq \{1, \ldots, k \}.
\end{equation} 
For a discussion of infrared divergences in the framework presented here, see \cite{fbmotic}. 

In our case divergences can only appear at zeroes of $\psi_G$ and we have Weinberg's theorem \cite{weinberg} which is a cornerstone for renormalisation theory.
\begin{prop}
Under the above conditions, the Feynman integral \eqref{fi} is convergent if and only if for all subgraphs $\gamma \subset G$ it holds that $s_\gamma<0$.
\end{prop}
Thus, a (sub-)graph $\gamma \subset G$ is called \textit{convergent} if $s_\gamma<0$ and \textit{divergent} if $s_\gamma \geq 0$. In the latter case $s_\gamma =0$ is referred to as a \textit{logarithmic (sub-)divergence} and $s_\gamma =1,2, \ldots$ as \textit{linear, quadratic etc.\ (sub-)divergences}.
\newline

The remarkable feature of perturbative quantum field theory and the reason for its success as a physical theory of interacting particles is the fact that, despite being ill-defined, the integrals $I_G$ still carry physical meaningful data. Renormalization is the art of extracting this data in a systematic way. In a nutshell\footnote{We do not want to dwell here on a precise definition of a physical meaningful renormalisation or its philosophical interpretation and refer the reader to the standard literature, e.g.\ \cite{iz}.}: The main approach to renormalise $I_G$ is to regularise the integral by adding a complex parameter $z\in \mb C$ and study $I_G(z)$ as a complex function. This allows to quantify the divergences of $I_G=I_G(z_0)$ in a mathematically sound way as poles in its Laurent expansion around $z_0$. Then one performs a renormalisation operation $\m R$ to render $I_G$ finite, i.e.\ to pass to the \textit{physical limit} $\lim_{z \rightarrow z_0}\m R ( I_G(z) )$. 

There are also methods without using an intermediate regulator, for example by 
\begin{enumerate}
 \item[-] modifying the integration domain $\sigma_G$ in order to shift it away from the singularities of the integrand \cite{bk}.
 \item[-] modifying the integrand $f_G$ in order to get rid of the singularities before integrating \cite{bk-as}. 
\end{enumerate}

The common feature of all of these methods is that they can be formulated as a rescaling of the physical constants in the given theory (in mathematical terms, the renormalisation procedure can be formulated as a special version of Birkhoff decomposition, cf.\ \cite{ck}). 
\newline

We demonstrate the latter approach in the case of at most logarithmic subdivergences. Let $G$ be a connected graph with only logarithmic subdivergences. Denote by $\m D=\{\gamma \subset G \mid s_\gamma =0\}$ the set of divergent subgraphs of $G$ and call $\m F\subset \m D$ \textit{a forest of} $G$ if 
 \begin{equation*}
  \text{ for all } \gamma, \eta \in \m F: \text{ either } \gamma \subset \eta \text{ or } \eta\subset \gamma \text{ or } \gamma \cap \eta = \emptyset.
 \end{equation*}
We want to define for every $\gamma \in \m D$ a subtraction on the integrand which eliminates the corresponding divergence of $f_G$. A naive definition term by term would not work though as one has to take the nestedness and possible overlaps of subdivergences into account. It turns out that forests of $G$ are the appropriate tool to organize this operation. Therefore, we define the \textit{renormalised Feynman integral} by Zimmermann's forest formula \cite{zman}
\begin{equation}\label{forfor}
 I^{\mathit{ren}}_G= \sum_{\text{forests }\m F }(-1)^{|\m F|} \int_{\sigma_G} f_{G,\m F} \nu_G
\end{equation}
where 
\begin{equation*}
f_{G,\m F}=(\psi_{G/\m F} \psi_{\m F})^{-\frac{d}{2}}\log \frac{   \phi_{G/\m F} \psi_{\m F}    +  \phi^0_{ \m F} \psi_{ G /\m F} }{   \phi^0_{G/\m F} \psi_{\m F}    +  \phi^0_{ \m F} \psi_{ G/\m F}  }
\end{equation*}
with $\psi_{\m F} :=  \psi_{  \m F '}$ and $\phi_{\m F} :=  \phi_{\m F'} $ for
\begin{equation*}
\m F' := \bigcup_{\gamma \in \m F} \big( \gamma / \bigcup_{ \eta \in \m F, \eta \subsetneq \gamma }\eta \big)
\end{equation*}
and the superscript $0$ in $\phi$ denotes evaluation at a fixed \textit{renormalisation point} $(p,m)=(p_0,m_0)$. 

For a proof that $I^{\mathit{ren}}_G$ is finite and a derivation of the general forest formula we refer to \cite{bk-as}. In the case of subdivergences of higher degree simple subtractions are not enough to render the integrand finite. One has to combine partial integrations (to reduce the degree of divergence) with subtractions of Taylor polynomials (to get rid of the resulting boundary terms) in order to renormalise the integrand. 
The formulae get considerably more complicated in this case but the overall structure does not change. The upshot is that renormalisation is still organized by the forest formula, and thus by a Hopf algebra, cf.\ \cite{ck} and Theorem $8$ in \cite{bk-as}.

\section{Moduli spaces of graphs}
\subsection{Outer space and moduli spaces of graphs} \label{osmsog}
Let us start with the definition of Outer space, as introduced by Culler and Vogtmann in \cite{cv}. Fix $n\in \mb N$ and call a graph $G$ \textit{admissible} if 
\begin{enumerate}
\item its rank or loop number $|G|=h_1(G)$ equals $n$,
 \item it is $1$PI or \textit{core} or \textit{bridgefree}; deleting an edge reduces its loop number,
 \item all (internal) vertices of $G$ have valence greater or equal to three.
\end{enumerate}

Let $R_n$ denote the rose graph with $n$ petals, i.e.\ the graph consisting of a single vertex and $n$ edges, and consider a space of triples $(G,g,\lambda)$ where $G$ is admissible, $g:G \rightarrow R_n$ a homotopy equivalence (called a \textit{marking}) and $\lambda$ a metric on $G$ that assigns to each $e \in G$ a positive length. Two elements $(G,g,\lambda),(H,h,\eta)$ are considered equivalent iff there is a homothety $\varphi$ between the metric spaces $(G,\lambda)$ and $(H,\eta)$, such that $h\circ \varphi$ is homotopic to $g$. This defines an equivalence relation on the space of all admissible marked metric graphs of rank $n$ and we denote the quotient, called \textit{(Culler-Vogtmann) Outer space}, by $CV_n$. 

There is a natural action of $Aut(F_n)$ on this space. An automorphism $\alpha$ acts on an equivalence class $[(G,g,\lambda)]$ by composing the map $g: G \rightarrow R_n$ with the homotopy equivalence $\ti \alpha: R_n \rightarrow R_n$ that is induced by identifying each (oriented) petal of $R_n$ with a generator of $F_n$. From the above notion of equivalence it follows that inner automorphisms act trivially, so that effectively it reduces to an action of $Out(F_n):=Aut(F_n)/Inn(F_n)$, the group of outer automorphisms of $F_n$.

As a topological space, $CV_n$ decomposes into a disjoint union of open simplices in the following way. For each marked graph $(G,g)$ consider the set of points obtained from changing the metric $\lambda$, i.e.\ by varying the edge lengths subject to the condition of positivity. By the equivalence of scaled metrics we can restrict to the case where each metric $\lambda$ on $G$ satisfies 
\begin{equation*}
vol_{\lambda}(G):= \sum_{e \in G} \lambda(e) =1.
\end{equation*} 
Hence, the space of allowed metrics on $(G,g)$ parametrises the interior of an $(|E(G)|-1)$-dimensional simplex $\Delta_G$. A face of $\Delta_G$ lies in $CV_n$ iff the edge set of $G$ on which $\lambda$ vanishes forms a forest in $G$. Vice versa, missing faces correspond to metrics vanishing on subgraphs $\gamma \subset G$ with $|\gamma|>0$. Elements of these faces are called \textit{points at infinity}. 
\newline

The whole construction naturally generalizes to the case of graphs with $k$ additional basepoints. These basepoints can be thought of as external edges in the sense of Definition \ref{graphs}. In this case one considers labeled graphs $(G,\{v_1,\ldots, v_k\})$, markings become homotopy equivalences $g: (G,\{v_1,\ldots, v_k\}) \rightarrow (R_n,\{v\}) $ and two labeled and marked metric graphs are considered to be equivalent if there is a homothety $\varphi: (G,\{v_1,\ldots, v_k\}) \rightarrow (H,\{w_1,\ldots, w_k\})$ such that $h\circ \varphi \simeq g \text{ rel } \{v_1,\ldots,v_k\}$. The resulting spaces are denoted by $CV_{n,k}$.

For $k=0$ one recovers the definition of Outer space. The case $k=1$ is called \textit{Autre or Auter space}. It allows to study the full automorphism group $Aut(F_n)$  as the existence of a basepoint makes the action of inner automorphism nontrivial. For $k\geq 2$ one obtains spaces equipped with actions of the groups $Out(n,k)\cong F_n^{k-1} \rtimes Aut(F_n)$, see \cite{chkv}. 
\newline

The general idea behind all these constructions is to have nice spaces on which these groups act, allowing to study them using geometric and topological tools. A special role is then played by the corresponding orbit space, the quotient
\begin{equation*}
 MG_{n,k}:= CV_{n,k}/Out(n,k),
\end{equation*}
the \textit{moduli space of rank $n$ metric graphs with $k$ external edges}.

\subsection{A moduli space of coloured graphs}
Unfortunately, the description of $CV_{n,k}$ as disjoint union of open simplices does not quite survive the projection onto $MG_{n,k}$. Indeed, under the quotient operation some open simplices get folded onto themselves. Heuristically speaking, this is due to the fact that without the marking, multi-edges between two vertices become indistinguishable. 

Although both graph polynomials $\psi$ and $\phi$ respect this symmetry as they are invariant under the corresponding permutations of edge-variables, it will be more convenient to work on an intermediate moduli space of coloured graphs. We therefore consider in the following graphs with their internal edges coloured by injective maps $c:E(G) \rightarrow \{1, \ldots, 3(n-1)+k \}$\footnote{An admissible graph of rank $n$ with $k$ legs can have at most $3(n-1)+k$ internal edges.}.
From a physics viewpoint, the colours play the role of placeholders for external data such as particle types and masses (determining the explicit form of the Feynman integrand $f_G$). Mathematically, they serve as fixed coordinates on the edges of $G$, thereby removing the above described symmetry under permutations of multi-edges. Therefore, the resulting moduli space of coloured graphs will behave combinatorially like a finite version of $CV_{n,k}$.  

\begin{defn}
Fix $n,k \in \mb N$ and let $C=C(n,k):=\{1, \ldots, 3(n-1)+k\}$. The moduli space of rank $n$ metric \textit{holocoloured} graphs with $k$ external edges is defined as
 \begin{equation*}
X_{n,k}:= \{ (G,\lambda,c) \mid \lambda: E(G) \longrightarrow \mb R_+,\ c:E(G) \longrightarrow C \}_{/ \sim}  
 \end{equation*}
where $G$ is admissible with $|G|=n$, has $k$ legs and every internal edge is coloured differently using $C$ as set of colours. The equivalence relation $\sim$ is given by $(G,\lambda) \sim (H,\eta)$ if there is a colour-respecting homothety $\varphi:G \rightarrow H$ such that $\lambda = \eta \circ \varphi $.
\end{defn}

\begin{rem}
 In principle one could further restrict the set of admissible graphs by bounding the allowed vertex valency from above. This would produce more missing faces in the resulting moduli spaces. Such spaces make sense for realistic Feynman amplitudes, but for the toy model presented here we simply consider the most general case.  
\end{rem}

The upshot is that $X_{n,k}$ decomposes into a finite disjoint union of open simplices, one for each admissible coloured graph, analogous to the description given for $CV_{n,k}$ in the previous section.
\newline

A convenient bookkeeper for the face relations in $X_{n,k}$ is the set of of all rank $n$ holocoloured graphs with $k$ legs, partially ordered by 
\begin{equation*}
 (G,c) \leq (G',c') \Longleftrightarrow \exists \text{ forest } F \subset G' : G'/ F=G \  \wedge c= c' \rest{E(G')\setminus E(F)}.
\end{equation*}
 Equivalently, it is the set of all open simplices in $X_{n,k}$ partially ordered by face relations. We denote this poset by $\m X_{n,k}$. Its colourless variant (or rather its geometric realization) plays a prominent helpful role in the study of the groups $Out(n,k)$.

\begin{rem}
 The symmetric group $S_{3(n-1)+k}\cong Perm(C) =:\Sigma_C $ acts on $X_{n,k}$ by changing the colours, $\sigma.(G,c):=(G,\sigma\circ c)$, and we retrieve the moduli space of metric graphs $MG_{n,k}$ as the orbit space of this action, $X_{n,k} / \Sigma_C = MG_{n,k}$.
\end{rem}

\subsection{A compactification of $X_{n,k}$}\label{compac}
We describe a compactification of $X_{n,k}$ following the work of \cite{bv} and \cite{v-bord} for Outer space. The construction will not depend on the colouring, so we drop it from the notation temporarily. 
\newline

Faces at infinity in $X_{n,k}$ correspond to degenerate metrics in the following sense. Let 
\begin{equation*}
\dot \Delta_G=\{ (x_1, \ldots, x_N) \mid \sum x_i=1, x_i > 0 \}, \ N=|E(G)|,
\end{equation*}
denote an open simplex in $X_{n,k}$ associated to an admissible coloured graph $G$ of rank $n$ with $k$ legs. In this standard parametrisation every face in the boundary of $\dot \Delta_G$ is described by a set of vanishing coordinates, or equivalently, by a set $S\subset G$ of zero-length edges in $G$. Such a face is thus an element of $X_{n,k}$ iff the graph $G/S$ is still of rank $n$. This is the case iff $S$ is a forest in $G$. We conclude that faces at infinity in $X_{n,k}$ correspond to pairs $(G,\gamma)$ where $G$ is admissible and $\gamma\subset G$ is a subgraph of $G$ with $|\gamma|>0$.

 \begin{example}
 Consider the Dunce's cap graph of Figure \ref{dunce} as an element of $\m X_{2,4}$ (coloured by $\{1,2,3,4\}$). The set of postive metrics of volume one on $G$ describe an open cell, its closure in $X_{2,4}$ is depicted in gray in Figure \ref{cell}. Faces in red correspond to metrics vanishing on subgraphs $\gamma \subset G$ with $|\gamma|>0$, hence lie at infinity in $X_{2,4}$.
    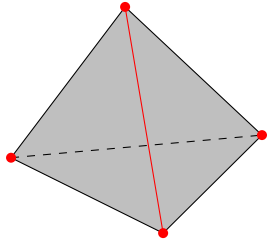
\begin{figure}[!htb]
  \begin{tikzpicture}
   \coordinate (v0) at (0,0);
   \coordinate  (v1) at (2,-1); 
   \coordinate (v2) at (1.5,2); 
   \coordinate (v3) at (3.3,0.3);
   \fill[fill=lightgray] (v0) -- (v1) -- (v2) -- (v0);
   \fill[fill=lightgray] (v1) -- (v2) -- (v3) -- (v1);
   \fill[fill=lightgray] (v0) -- (v1) -- (v2) -- (v3);
   \fill[fill=lightgray] (v0) -- (v1) -- (v2) -- (v3);
   \draw (v0) -- (v1);
   \draw (v0) -- (v2);
   \draw[dashed] (v0) -- (v3);
   \draw (v1) -- (v3);
   \draw (v2) -- (v3);
   \draw[red] (v1) -- (v2);
    \fill[red] (v0) circle (.066cm);
   \fill[red] (v1) circle (.066cm);
   \fill[red] (v2) circle (.066cm);
   \fill[red] (v3) circle (.066cm);
  \end{tikzpicture}
\caption{A cell in $MG_{2,4}$}\label{cell}
 \end{figure}
\end{example}

To construct a compactification of $X_{n,k}$ we proceed simplex by simplex using a method analogous to the Borel-Serre construction for arithmetic groups. From now on denote by $\sigma_G$ always a relatively closed simplex in $X_{n,k}$, i.e.\ $\sigma_G$ is the open simplex $\dot \Delta_G$ together with all of its faces that correspond to graphs $G/F$ where $F\subset G$ is a forest in $G$. 

Consider a point at infinity $x \in \overline{\sigma_G}$ where a subset of edge variables $S\subset G$ vanishes. We can restrict our attention to the case where $S=G_1$ is a $1$PI or core subgraph of $G$ - setting the remaining edge lengths in $S \setminus G_1$ to zero describes a face of $\sigma_{G_1}$ (which does not lie at infinity). 

Possible directions of approaching $x$ correspond to flags of subgraphs of $G$ in the following way. The set of metrics on $G_1$ defines, after rescaling, a new simplex $\sigma_{G_1}$. If a metric vanishes on another core subgraph $G_2 \subset G_1$, we can repeat this construction to obtain a simplex $\sigma_{G_2}$, and so on. This process ends after a finite number of steps since the loop number of the graphs considered must decrease in every step, $|G_i|>|G_{i+1}|$. 
A point at infinity in $X_{n,k}$ can thus be described by a finite sequence of core subgraphs, a \textit{flag} $G=G_0\supset G_1 \supset G_2 \supset \ldots \supset G_m$, each equipped with a metric on its edges, normalized to volume one.  

For any core subgraph $\gamma \subset G$ there is a projection map $r_\gamma: \sigma_G \rightarrow \sigma_\gamma$. It is defined by restricting a metric on $G$ to $\gamma$ and rescaling it to volume one, thereby defining a point in $\sigma_\gamma$. The product of these maps forms a composite map
\begin{equation*}
 r: \sigma_G \longrightarrow \prod_{\gamma \subset G \text{ core}} \overline{\sigma_{\gamma}}
\end{equation*}
which is an embedding (here $G$ is counted as a core subgraph of itself). The \textit{compactified cell} $\ti \sigma_G$ is defined as the closure of the image of $r$,
\begin{equation*}
 \ti \sigma_G := \overline{ r(\sigma_G)}.
\end{equation*}

\textit{Alternative description of $\ti \sigma_G$ (cf.\ \cite{bek,bk,bk-as})}: Another way of parametrising the standard $n$-dimensional simplex is to describe it as subset of $n$-dimensional real projective space,
\begin{equation*}
 \Delta^n= \{ [x_0: x_1: \ldots :x_n] \mid x_i\geq 0 \} \subset \mb P(\mb R^{n+1}).
\end{equation*}
In this projective setup let $\mb P_G=\mb P(\mb R^N)$. Then we can rephrase the previous discussion as follows. The compactified cell $ \ti \sigma_G$ is the subset of $\prod_{\gamma} \mb P_{\gamma}$ obtained from $\overline{\sigma_G}$ by a sequence of blowups along the (strict transforms of) subspaces 
\begin{equation*}
 L_{\gamma}= \{ x_e= 0 \mid e \in \gamma \} \subset \overline{\sigma_G}=\{ [x_1: \ldots :x_N] \mid x_i\geq 0 \}
\end{equation*}
where each $\gamma$ is a proper core subgraph of $G$. The sequence of blowups proceeds along subspaces of increasing dimension, so it is determined by the inclusion relation on subgraphs whereas for disjoint subgraphs the order does not matter. We recover thus the above description of points at infinity by flags of core subgraphs of $G$.

\begin{prop}
 Both constructions are equivalent, i.e.\ for every admissible graph $G$ both compactified cells are isomorphic (as smooth varieties).
\end{prop}
\begin{proof}
 The projective simplex $\Delta^n_p$ is isomorphic to the standard one $\Delta_s^n$ via the regular map
\begin{equation*}
 \zeta: \Delta_p^n \longrightarrow \Delta^n_s, \quad [x_0 : \ldots : x_n] \longmapsto \frac{1}{x_0+ \ldots + x_n}(x_0, \ldots, x_n).
\end{equation*}
Under this map the family $\{L_\gamma \mid \gamma \subset G \text{ core} \}$ transforms into a linear subspace arrangement in $\mb R^{N_G-1}$. The compactified cell $\ti \sigma_G$ is a wonderful model for this arrangement in the sense of DeConcini-Procesi \cite{dp}. More precisely, it is the wonderful model for the maximal building set $\m B=\{\zeta(L_\gamma) \cap \Delta_s^{N_G-1} \mid \gamma \subset G \text{ core}\}$. The results in \cite{dp} show that both descriptions of $\ti \sigma_G$ are equivalent. Moreover, the construction through a sequence of blowups provides local coordinates\footnote{See Section \ref{reonco}.} on this wonderful model using the notion of nested sets which here are given by totally ordered subsets of $\m B$, hence by flags of core subgraphs of $G$. 
\end{proof}

By construction the projection map $\beta: \ti \sigma_G \rightarrow \overline{\sigma_G}$ is an isomorphism outside of the exceptional divisor 
\begin{equation*}
 \m E=\m E_G:=\beta^{-1}( \cup_{\gamma} L_{\gamma} ) = \cup_\gamma \m E_\gamma, \quad  \m E_\gamma:=  \beta^{-1}(L_{\gamma})\cong \mb P_\gamma \times \mb P_{G/\gamma},
\end{equation*}
with its inverse given by the map $r$. Therefore it makes sense to call the elements in $\m E_\gamma \subset \m E_{G}$ the \textit{new faces} of $\ti \sigma_G$. In a graphical notation that will be useful later we write for a new face $\tau \subset \ti \sigma_G$, corresponding to the blowup of a $L_{\gamma}$ or to an intersection of multiple such faces,
\begin{equation}\label{newface}
 \tau \sim (G,\m F) \sim (G=G_0\supset G_1 \supset G_2 \supset \ldots \supset G_m),
\end{equation}
where $\m F$ is a flag $G_1 \supset G_2 \supset \ldots \supset G_m$ of core subgraphs of $G$. If $m=1$, then the pair $(G,\m F)=(G,G_1)$ describes a maximal new face of $\ti \sigma_G$; if $m>1$, then $(G,\m F)$ describes the intersection of the faces $(G,G_i)$ for $i = 1, \ldots, m$. 

All other faces of $\ti \sigma_G$ have a description that is induced by the face relation in $X_{n,k}$.

\begin{lem}\label{boundary}
 If $G'=G/F$ for a forest $F \subset G$, then $\ti \sigma_{G'}$ is a face of $\ti \sigma_{G}$.
 \end{lem}
 
 \begin{proof}
 The face relation in $X_{n,k}$ via contraction of forests $F\subset G$ defines a map $\pi: \sigma_{G} \rightarrow \sigma_{G'}$. Because of the property 
\begin{equation*}
\gamma \subset G \text{ is core } \Longrightarrow  \gamma/ (F\cap \gamma) \subset G' \text{ is core, } 
\end{equation*}
$\pi$ lifts to a map $\ti \pi: \ti \sigma_{G} \rightarrow \ti \sigma_{G'}$ such that
\[
  \begin{tikzcd}
  \ti  \sigma_{G} \arrow{r}{\ti \pi} \arrow[swap]{d}{\beta} & \ti \sigma_{G'} \arrow{d}{\beta'} \\
    \sigma_{G} \arrow{r}{\pi}  & \sigma_{G'} \
  \end{tikzcd}
 \] commutes. 
 \end{proof}

\begin{lem}\label{vertex}
 A vertex $\nu$ of $\ti \sigma_G$ is described by $\nu \sim (G,T,e_0,\ldots,e_{n-1})$ where $T$ is a spanning tree $T\subset G$ and $(e_1,\ldots,e_n)$ an ordering of the edges in $G \setminus T$.
\end{lem}
\begin{proof}
 Lemma \ref{boundary} shows that the contraction of a spanning tree $T\subset G$ defines a facet of $\ti \sigma_G$ of maximal codimension, parametrised by the edge variables of $G/ T$ which is a rose graph with $n$ petals. By \eqref{newface} the vertices of $\ti \sigma_{G/T}$ are given by maximal flags of core subgraphs of $G/T$ which can be represented by orderings of the $n$ petals of $G/T$.
\end{proof}

\begin{cor}
 Each cell $\ti \sigma_G$ is the convex hull of its vertices, i.e.\ a polytope.
\end{cor}

Finally, we define the compactification of $X_{n,k}$ as the result of gluing together all the cells $\ti \sigma_{(G,c)}$ along their common boundaries.
\begin{defn}
 The compactified moduli space of admissible rank $n$ metric holocoloured graphs with $k$ external edges is 
 \begin{equation*}
  \ti X_{n,k}:= \big( \cupdot_{(G,c) \in \m X_{n,k} } \ti \sigma_{(G,c)} \big)_{/ \sim },
 \end{equation*}
where the relation $\sim$ is induced by the face relation on $X_{n,k}$ through the maps $\ti \pi$ constructed in Lemma \ref{boundary}.
\end{defn}

\begin{example}
Figure \ref{compcell} shows the compactified cell $\ti \sigma_G$ for the Dunce's cap graph. The dark gray faces are the results of blowing up $\Delta_G$ first along the $0$-faces $\{x_i=1\}$ for $i =3,4$, and then along the $1$-face $\{x_3=x_4=0\}$ corresponding to the core subgraph $\gamma \subset G$ formed by edges $3$ and $4$.
  
  \begin{figure}[!htb]
  \begin{tikzpicture}
   \coordinate (v0) at (0,0);
   \coordinate (v01) at (0,-0.6);
   \coordinate (v02) at (0.4,-0.25);
   
   \coordinate  (v1) at (1.5,1.6); 
    \coordinate (v11) at (2,1.66);
      
   \coordinate (v2) at (1.6,-1.1);
   \coordinate (v21) at (2.4,-1);
      
   \coordinate (v3) at (3,0);
   \coordinate (v31) at (3.3,0.2);
   \coordinate (v32) at (3.3,-0.2);
   
    \fill[fill=lightgray] (v0) -- (v01) -- (v2) -- (v21) -- (v11) -- (v1) -- (v0);
   \fill[fill=lightgray] (v2) -- (v21) -- (v32) -- (v31) -- (v11);
 
   \fill[fill=gray,opacity=0.66] (v0) -- (v01) -- (v02) -- (v0);
   \fill[fill=gray,opacity=0.66] (v3) -- (v31) -- (v32) -- (v3);
   \draw[dashed] (v0) -- (v02);
   \draw[dashed] (v01) -- (v02);
   \draw[dashed] (v32) -- (v3);
   \draw[dashed] (v31) -- (v3);
   
   \fill[fill=gray] (v1) -- (v11) -- (v21) -- (v2) -- (v1);
   
   \draw (v0) -- (v01);
   \draw (v01) -- (v2);
   \draw (v2) -- (v21);
   \draw (v21) -- (v11);
   \draw (v0) -- (v1);
   \draw (v1) -- (v11);
   \draw (v11) -- (v31);
   \draw (v31) -- (v32);
   \draw (v32) -- (v21);
   \draw (v21) -- (v2); 
   \draw (v2)--(v1);
   \draw[dashed] (v02) -- (v3);
  \end{tikzpicture}
\caption{A compactified cell in $\ti X_{2,4}$}\label{compcell}
 \end{figure}
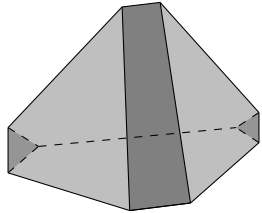
\end{example}

\section{Pseudo complexes and piecewise distributions}\label{pscpwd}
This section introduces some notions necessary to formulate and renormalise Feynman amplitudes on moduli spaces of graphs. For an introduction to simplicial, $\Delta$- and CW-complexes, and a detailed discussion of the differences between these notions, see \cite{ah}. For an introduction to distributions see \cite{gs} or \cite{ho}.
\newline

As discussed in Section \ref{compac}, the spaces $X_{n,k}$ are not real $\Delta$-complexes (also known as \textit{semi-simplicial complexes}), but have missing faces. Let us call such spaces pseudo $\Delta$-complexes.
\begin{defn}
A topological space $K$ is a pseudo $\Delta$-complex iff $K = L \setminus F $ for a finite $\Delta$-complex $L$ and a $F$ a subcomplex of $L$. Face relations in $K$ are then naturally inherited from $L$. Equivalently, we say $K$ is pseudo iff it is the disjoint union of finitely many open simplices modulo face relations.

We call the elements of $K$ pseudo simplices, i.e.\ $\sigma$ is a pseudo simplex in $K$ if there is $\tilde \sigma \in L$ such that $\sigma= \tilde \sigma \setminus (\cup_{\tau \in F} \tilde \sigma \cap \tau)$. 
\end{defn}

\begin{example} All spaces $X_{n,k}$ with $n\geq 2$ are pseudo simplicial complexes - for instance, all $0$-simplices are missing - whereas all $X_{1,k}$ are ordinary $\Delta$-complexes - every edge contraction transforming $G\in \m X_{1,k}$ to a coloured rose is allowed.
\end{example}

\begin{rem}
In the simplicial category pairs $(L,F)$ are known as \textit{relative simplical complexes}, introduced by Stanley in \cite{sta87}. Since we are dealing with more general spaces, we will avoid confusion (for instance with the concept of \textit{relative CW complexes}) by sticking to the term \textit{pseudo}.
\end{rem}

Every pseudo simplex is locally just a manifold with corners. Therefore, differential forms and integration can be defined on these objects\footnote{See \cite{grmo} for applications of differential forms on simplicial complexes.}. Moreover, simplices are orientable so that we have a natural identification of distribution densities and volume forms \cite{ni}. This allows to define distributions on (pseudo) complexes. To formulate amplitudes as evaluations of distributions over the spaces $X_{n,k}$ we need to take into account contributions from all of their pieces. Therefore, we have to integrate over lower dimensional simplices as well, contributions that are not taken care of by the usual theory of integration. To cope with this anomaly, we change the definition of a distribution slightly from the usual one. 

\begin{defn}
 A piecewise distribution on a (pseudo) complex $K$ is a collection $u=\{ u_{\sigma} \mid \sigma \in K \}$ of distributions, one for each of its (pseudo) simplices. The value of $u$ at a test function $\varphi \in C_c^{\infty}(K)$, denoted by $\langle u \mid \varphi \rangle$, is given by the sum over all its single contributions
 \begin{equation*}
  \langle u \mid \varphi \rangle := \sum_{\sigma \in K}\langle u_{\sigma} \mid \varphi \rest \sigma \rangle.
 \end{equation*}
 
 A piecewise distribution $u$ is said to respect face relations if the following holds. If $\tau$ is a face of $\sigma \in K$, then $u_{\tau}= {u_{\sigma}}_{|\tau}$, where the restriction of a (regular) distribution $u$ to a submanifold $S$ is defined by (cf.\ \cite{gs})
 \begin{equation}\label{restr}
 S \overset{loc.}=\{x_1,\ldots, x_k = 0 \} \Longrightarrow \langle u\rest S \mid \varphi \rangle = \int_S \varphi(x)u(x)\prod_{i=1}^k\delta(x_i) dx_1 \cdots dx_n.
 \end{equation}
\end{defn}

\begin{rem}
\begin{enumerate}
\item A piecewise distribution $u$ on a $\Delta$-complex $K$ of dimension $d$ defines cochains $u^i\in C^i(K,\mb R)$ for $i=0, \ldots, d$ by
\begin{equation*}
\sigma=\sum_j a_j \sigma_j \longmapsto u^i(\sigma):=\sum_j a_j \langle u_{\sigma_j} \mid \chi \rest{\sigma_j}  \rangle, \text{ with } \chi \equiv 1 \in C_c^{\infty}(K).
\end{equation*}
If $K$ is pure and $u$ respects face relations, then it is completely determined by its values on the facets of $K$. In this case we have a sequence 
\begin{equation*}
C^{d}(K,\mb R) \ni u^{d} \longrightarrow u^{d-1} \longrightarrow \ldots \longrightarrow u^{0} \in C^{0}(K,\mb R)
\end{equation*}
that is induced by the restriction map and carries essentially the same information as the coboundary operation on $C^*(K,\mb R)$.
\item The definition of piecewise distributions also works for more general spaces, such as polytopal or CW-complexes and even stratified spaces. The only important property needed is a notion of integration on each building block compatible with the corresponding boundary or face relations. In every such setting distributions, and even differential forms and currents, can be defined as above. 
\end{enumerate}
\end{rem}

\section{Feynman amplitudes as piecewise distributions}\label{faapwd}
From now on let $n,k$ be fixed. To minimize notation we denote by $X=X_{n,k}$ the moduli space of metric holocoloured graphs of rank $n$ with $k$ legs. Furthermore, in the following we write simply $G$ for a coloured graph $(G,c)\in \m X$. 
\newline

Assume given a set of masses $m_c$, one for each colour $c\in C=\{1,\ldots, 3(n-1)+k \}$, and let $p\in (\mb R^{d})^k$ denote a generic external momentum configuration. Inspecting the Feynman integrand $f_G$ we conclude from \eqref{edgeco} that the graph polynomials respect face relations, i.e.\ for each coloured edge $e \in E(G)$ we have
\begin{equation*}
 \psi_{G / e}  = \psi \rest {x_e=0} , \quad  \phi_{G / e}  = \phi \rest{x_e=0}, \quad \Xi_{G/e}=\Xi \rest {x_e=0}.
\end{equation*}
On the other hand, $f_G$ depends also on $s_G$, the superficial degree of divergence of $G$ through the exponent of $\psi_G/\Xi_G$. It is a discontinuous function on $\sigma_G$, given by
\begin{equation*}
s_G=d|G|-2N = dn - 2\big( \sum_{e \in E(G)} 2\theta(x_e)-1 \big) 
\end{equation*} 
with $\theta$ the Heaviside distribution
\begin{equation*}
\theta(x)=\begin{cases}0 & \text{ if } x<0 \\
           \frac{1}{2} & \text{ if } x=0 \\
           1 & \text{ else. }
          \end{cases}
\end{equation*}
Therefore, $f_G$ respects face relations, although in a discontinuous way, so that we are naturally led to work within the class of piecewise distributions on $X$.
\begin{defn}
The Feynman piecewise distribution on $X$ is defined as the collection 
\begin{equation*}
 t=t(p)=\{t_G \mid G \in \m X\}, \quad t_G: C_c^{\infty}(\sigma_G) \longrightarrow \mb C,\ \varphi \longmapsto \langle t_G \mid \varphi \rangle:= \int_{\sigma_G}\varphi \omega_G,
\end{equation*}
with $\omega_G=f_G(p,c)\nu_G$ as defined in \eqref{omega}. 
\end{defn}
This definition is justified by the following
\begin{prop}
 $t$ is a piecewise distribution on $C_c^{\infty}(X)$ that respects face relations.
\end{prop}
\begin{proof}
 First, we show that each $t_G$ is a distribution on $\sigma_G$. Let $\varphi \in C_c^{\infty}(\sigma_G)$. Recall the discussion of divergences of the Feynman integral $I_G$ in Section \ref{renorm}. As $\varphi$ is compactly supported, it cannot meet the divergent locus of $G$ which is contained in the missing faces of $\sigma_G$. Hence, $\langle t_G \mid \varphi \rangle$ is well-defined for all $ \varphi \in C^{\infty}_c(\sigma_G)$. 
 
 Linearity of $T_G$ is clear. To probe the continuity of this map, let $\varphi_k$ be a sequence in $C^{\infty}_c(\sigma_G)$ converging to a test function $\varphi$. This means, there is a compact subset $K \subset \sigma_G$ with supp$(\varphi_k) \subset K$ for all $k \in \mb N$ and $\varphi_k \rightarrow \varphi$ uniformly on $K$. Since $\omega_G$ is a smooth differential form away from the divergent locus of $G$, all products $\varphi_k \omega_G$ are compactly supported differential forms, converging to $\varphi \omega_G$. More precisely, this holds on the interior of $\sigma_G$, but we can neglect the discontinuity of $f_G$ at $\sigma_G \setminus \dot \sigma_G$ since it is still bounded there and thus not seen by $\dim(\sigma_G)$-dimensional integration. 
We conclude
\begin{equation*}
 \langle t_G \mid \varphi_k \rangle = \int_{\sigma_G} \varphi_k \omega_G = \int_{K} \varphi_k \omega_G \longrightarrow \int_K \varphi \omega_G = \int_{\sigma_G} \varphi \omega_G=  \langle t_G \mid \varphi \rangle.
\end{equation*}

It remains to check that $t$ is compatible with face relations. Let $\sigma_\gamma \subset \sigma_G$, where $\gamma$ is obtained from $G$ by contraction of a forest $F \subset G$, i.e.\ $\sigma_\gamma$ is the subset of $\sigma_G$ where all edge variables associated to $F$ are set to zero. Since each $t_G$ is a regular distribution, its restriction to $\sigma_\gamma$ is given by \eqref{restr}, 
\begin{equation*}
 {t_G} \rest {\sigma_\gamma}: C_c^{\infty}(\sigma_\gamma) \ni \varphi \longmapsto \int_{\sigma_G} \varphi \delta(\sigma_\gamma) \omega_G,
\end{equation*}
where integration against $\delta(\sigma_\gamma)$ evaluates the integrand at all edge variables of $F$ set to zero. Therefore, this integral equals
\begin{equation*}
\int_{\sigma_G} \varphi \delta(\sigma_\gamma) \psi_G^{-\frac{d}{2}} \left(\frac{\psi_G}{\Xi_G}\right)^{-\frac{s_G}{2}} \nu_G = \int_{\sigma_\gamma} \varphi \left( \psi_G^{-\frac{d}{2}} \left(\frac{\psi_G}{\Xi_G}\right)^{-\frac{s_G}{2}} \right)\rest {x_e=0,  e \in F} \nu_\gamma  = \int_{\sigma_\gamma} \varphi \omega_\gamma ,
\end{equation*}
because both graph polynomials and the superficial degree of divergence respect face relations. The identity
\begin{equation*}
 \int \varphi \delta( \sigma_\gamma ) \nu_G = \int \varphi \rest {\sigma_\gamma} \nu_\gamma
\end{equation*}
follows from the definition of the restriction map on distributions using local coordinates.
\end{proof}

The previous discussion shows that for every generic momentum configuration $p$ we have a piecewise distribution $t(p)$ on $X$. Eventually we are interested in the value of $t(p)$ on the constant function $ \chi \equiv 1$, the \textit{(unrenormalised) Feynman amplitude} $\m A_n$ (\textit{of order $n$ at $p$}),
\begin{equation*}
\m A_n: p \longmapsto \langle t(p)\mid \chi \rangle := \sum_{\sigma \in X} \langle t_{\sigma}(p) \mid \chi \rest \sigma \rangle =\sum_{G \in \m X} \langle t_{G}(p) \mid \chi \rest {\sigma_G} \rangle.
\end{equation*}
Thus, renormalisation translates in this picture into the task of finding a well-defined expression for the limit $\lim_{k \rightarrow \infty} \langle t(p) \mid \varphi_k \rangle $ where $\varphi_k$ is a sequence of test functions converging to the characteristic function $\chi$ of the space $X$. 

\begin{rem}\label{rem6}
 \begin{enumerate}
  \item This is the algebraic geometer's definition of an amplitude as a projective integral. For a comparison to its ''real world'' version and a derivation of the latter see \cite{bk-as}. The constructions presented here work equally well in this case. It is important to note though that if $s_G=0$, then 
  \begin{equation*}
  f_G(p,m)=\psi_G^{-\frac{d}{2}} (\frac{\psi_G}{\Xi_G(p,m)})^{-\frac{s_G}{2}}= \psi_G^{-\frac{d}{2}}
  \end{equation*}
does not depend on $(p,m)$. The amplitude however does, by a multiplicative factor coming from renormalisation of the overall divergence of $G$.  
  \item By definition, this amplitude sums over all possible distributions of masses in graphs. If two or more masses are equal, these multiplicities can be taken into account with the help of appropriate symmetry factors.
  \item As mentioned above, for realistic field theories the valency of each vertex in an admissible graph is bounded also from above, so that one has to restrict the definition of a Feynman amplitude to an appropriate subcomplex of $X$. 
 \end{enumerate} 
\end{rem}

\section{Renormalisation on the compactification $\ti X_{n,k}$}\label{reonco}

To find a renormalised version of the Feynman distribution $t$ we will use the compactification $\beta:\ti X \rightarrow X$. Together with the pullback and pushforward operations on distributions this allows to study and control the behaviour of the divergent parts of each $t_G$. 

As we have seen above, the compactification $\ti X$ is a polytopal complex.
\begin{defn}
A polytopal complex $P$ is a (finite) collection of polytopes such that
\begin{enumerate}
\item if $q$ is a face of $p\in P$, then $q \in P$.
\item if $p,p' \in P$ and $p\cap p' \neq \emptyset$, then $p\cap p'= q \in P$.
\end{enumerate}
\end{defn}

All of the theory introduced in Section \ref{pscpwd} works also in the case of polytopal complexes (by definition and also because every polytopal complex can be triangulated into a simplicial complex). This allows to view $\beta^* t$ as a piecewise distribution on the compactification of $X$. Then, working on each polytope $\ti \sigma$ separately, we will find its divergent loci at the new faces of $\ti \sigma$ that are indexed by divergent subgraphs. It is important to note that the ultraviolet divergences we are considering here are independent of the colouring of $G$ and its subgraphs; as long as we are dealing with generic external momentum configurations or strictly positive masses, all possible divergences depend only on the topology of $G$ as uncoloured graph.

After all of these divergent loci and the corresponding poles are identified, we will employ the necessary subtraction operations to render each distributional piece of $\beta^* t$ finite. In the end we show that the so obtained distributions fit together in order to produce a piecewise distribution on $\ti X$ (or $X$ after pushing it back down along $\beta$) that has a finite value when evaluated at $\chi$. The result of this whole operation is then called the \textit{renormalised Feynman amplitude} $\m A_n^{\mathit{ren}}$.

As mentioned in the introduction this is nothing new, but merely a reformulation of renormalised parametric Feynman integrals in the context of moduli spaces of graphs. This problem has been studied long ago and solutions are well understood. Finite renormalised expressions for $I_G$ are given by various methods in the literature, see for instance \cite{bk,bk-as} or \cite{iz}.
\newline

Let $\ti \sigma$ denote the compactification of a pseudo simplex $\sigma$ as defined in Section \ref{compac} and $\beta: \ti \sigma \rightarrow \ol \sigma$ the corresponding projection. Recall that $\sigma$ corresponds to a holocoloured graph $G\in \m X$, where the colouring determines the explicit form of the graph polynomial $\Xi_G$. On the other hand, the shape of $\ti \sigma$ depends only on the topology of the graph $G$. 

Since distributions can be pulled back along smooth submersions, we have away from of the exceptional divisor $\m E$ 
\begin{equation*}
  \quad \langle \beta^* t_\sigma \mid \varphi \rangle = \langle t_\sigma \circ \beta \mid \varphi \rangle .
\end{equation*}
Of course, this is defined so far only for $\varphi$ compactly supported on the complement of $\m E$ in $\ti \sigma$; we have not taken care of the divergences yet. 
Moreover, note that no information is lost if we work on the compactification $\ti X$ because for a graph $G$, free of subdivergences, we have
\begin{equation*}
\forall \varphi \in C^{\infty}_c(\sigma):\  \int_{\ti \sigma_G} \beta^* (\varphi \omega_G ) = \int_{\sigma_G} \varphi \omega_G .
\end{equation*}
This follows from the fact that $\beta$ is a smooth isomorphism outside of $\m E$ which is a set of measure zero.

Thus, so far we have a collection $\tilde t= \{\beta^* t_{\sigma} \mid \sigma \in X\}$ that satisfies the properties of a piecewise distribution on $\ti X$, except we have not yet assigned an element of $\tilde t$ to every polytope in $\ti X$. The distributional pieces corresponding to polytopes in the exceptional divisor $\m E$ (and their faces) are still missing. Wherever it is defined, $\tilde t$ respects face relations, but due to the presence of divergences we cannot use these relations to determine the value of it at all new faces.

The workaround is to use a regularisation as explained in Section \ref{renorm}. To do so, we consider the constant $d$ in \eqref{superficial} as a complex parameter allowing us to work with finite intermediate piecewise distributions that we can pull back and extend on the whole space $\ti X$. Thus, for $d\in \mb C$ we define the \textit{regularised Feynman piecewise distribution} $t^d$ on $X$ by 
\begin{equation*}
 t^d=\{ t^d_\sigma \mid \sigma \in X \} = \{ t^d_G \mid G \in \m X \} \text{ with } \langle t^d_G \mid \varphi \rangle:= \int_{\sigma_G}\varphi f^d_{G} \nu_G.
\end{equation*}
Since $f_G$ is given by
\begin{equation*}
f_G= \psi_G^{-\frac{d}{2}}\left( \frac{\psi_G}{\Xi_G} \right)^{N-|G|\frac{d}{2}}= \psi_G^{-\frac{d}{2}}\left( \frac{\Xi_G}{\psi_G} \right)^{\frac{s_G}{2}}
\end{equation*}
and the superficial degree of divergence $s_G$ is bounded for admissible graphs $G \in \m X$, we can choose $d\in \mb C$ so that $t^d$ is a piecewise distribution on $X$ for which the pullback along $\beta$ delivers finite distributions on each $\ti \sigma$. Then $\tilde t^d$ can be extended to a piecewise distribution on the whole space $\ti X$. 

\begin{rem}
 The approach presented here is known as \textit{dimensional regularisation}. Of course there are many ways to regularise Feynman integrals. Two other methods regularly used are 
 \begin{itemize}
  \item[-] using a \textit{cut-off}, i.e.\ by reducing the integration domain $\ti \sigma_G$ to the complement of an $\epsilon$-neighborhood of the exceptional divisor $\m E$. In some sense, this idea is implicitly built into our formulation of Feynman integrals as distributions on moduli spaces of graphs; the support of a compactly supported test function on $X$ can not intersect with any divergent locus $L_\gamma$ for some $\gamma \subset G \in \m X$ with $|\gamma|>0$. However, using dimensional regularisation additionally provides more control to study the divergences of the distributions $t_G$ in detail. For an approach that focuses solely on modifying the integration domain $\sigma_G$ for regularisation and renormalisation see \cite{bk}.
  \item[-] \textit{analytic regularisation} which writes the Feynman integral as a \textit{Mellin transform} by replacing the integrand $f_G$ with
  \begin{equation*}
    \psi_G^{-\frac{d}{2}}\left( \frac{\Xi_G}{\psi_G} \right)^{\frac{s_G}{2}} \prod_{e\in E(G)}x_e^{a_e-1}.
  \end{equation*}
One shows that there is an open set $A\subset \mb C^N$ such that the integral converges for $(a_1, \ldots, a_N)\in A$. Renormalisation amounts then to find an analytic continuation to the point $(1,\ldots, 1)\in \mb C^N$.
 \end{itemize}
\end{rem}

Write
\begin{equation*}
 \tilde t^d= \{\tilde t^d_{\Sigma} \mid \Sigma \in \ti X \}
\end{equation*}
for the collection of distributions, one for each polytope of $\ti X$, where
\begin{equation*}
\tilde t^d_\Sigma := \begin{cases} \beta^* t^d_\sigma & \text{ if $\Sigma=\ti \sigma$ is the blowup of a $\sigma \subset X$,} \\
           (\beta^*t^d_\sigma) \rest{\Sigma}  & \text{ if $\Sigma\subset \m E_\sigma$ is a new face in the blowup of a $\sigma \subset X$.} \end{cases}
\end{equation*}

From now on we fix a (coloured) graph $G$ and consider the distribution $\ti t^d_G$ on the polytope $\ti \sigma_G$ in $\ti X$. For the loci of possible divergences, that is for faces $\Sigma \subset \m E_G \subset \ti \sigma_G$, we use the graphical notation introduced in Section \ref{compac}, Equation \eqref{newface}. There we have for every pair $(G,\m F)$ with $\m F$ a flag of core subgraphs of $G$ (with the induced colouring)
\begin{equation*}
  G=G_0 \supset G_1 \supset \ldots \supset G_m
\end{equation*}
a distribution $\tilde t^d_\Sigma$ such that
\begin{enumerate}
 \item[-] if $m=1$, i.e. $\m F=G_1$, then $\Sigma=\m E_{G_1}=\beta^{-1}(L_{G_1})$ and $\tilde t^d_\Sigma ={\beta^* t^d_G} \rest{\m E_{G_1}}$. 
 \item[-] if $m>1$, then $\tilde t^d_\Sigma$ is given by the restriction of $\beta^* t^d_G$ to the common face of the $\m E_{G_i}=\ti \sigma_{G_i}$ indexed by the flag $\m F = (G_1,\ldots ,G_m)$.
\end{enumerate}

These are the loci of possible divergences of $\tilde t^d_{\Sigma}$. The other faces of $\Sigma$ (i.e.\ intersections with faces corresponding to non-core subgraphs of $G$) do not carry additional information. 

On $\ti \sigma_G$ we use local coordinates to study the distribution $\tilde t_G^d$ in the vicinity of $\m E \subset \partial \ti \sigma_G$. Let $\tau \subset \m E$ be a new face, $\tau \sim (G,\m F)$ with $\m F = (G_1,\ldots ,G_m)$, and define
\begin{equation}\label{flagsforests}
 \gamma_m:=G_m,\ \gamma_{m-1}:=G_{m-1} / G_m \ \ldots \  \gamma_0:=G / G_1.
 \end{equation}
 In affine coordinates $x_{e}=1$\footnote{In the following we omit the subscript for evaluation at $x_e=1$.} for $e \subset \gamma_0 = G / G_1$ write $y_0$ for the vector of edge variables of $\gamma_0 \setminus e$ and $ y_i$ for those associated to $\gamma_i$ ($i=1, \ldots, m$),
 \begin{equation*}
 y_0 := (y_0^e)_{e \in E(\gamma_0)}, \quad y_i :=  (y_i^e)_{e \in E(\gamma_i)}.
 \end{equation*}
Furthermore, choose for every $i = 1,\ldots,m$ a single coordinate $y_i^*$ in $y_i$ and denote by $\mathbf{y}_i$ the vector $y_i$ with this coordinate set to $1$. Then, following the wonderful model construction of \cite{dp}, the blowup sequence from $\sigma$ to $\ti \sigma$ is locally described by the coordinate transformation $\rho=\rho_{\m F,(y_1^*,\ldots,y_n^*)}$ given by
\begin{equation} \label{locobu}
  \rho: ( y_0 ,y_1, \ldots , y_m ) \longmapsto 
  \big(y_0, y_1^* \cdot \mathbf{y}_1, (y_1^* y_2^*) \cdot \mathbf{y}_2, \ldots, (y_1^* \cdots y_m^*) \cdot \mathbf{y}_m   \big).
\end{equation}
Thus, in local coordinates we have $\beta=\rho$. In order to detect the poles of $\rho^*\omega_G$ along $\tau$ we deduce its scaling behaviour from the contraction-deletion relations for graph polynomials.

\begin{lem}\label{rhosca}
Let $\rho$ be given by \eqref{locobu}. The graph polynomials $\psi_G$ and $\phi_G$ satisfy
\begin{equation*}
 \big(\psi_G \circ \rho\big)(y)= \prod_{i=1}^m (y_i^*)^{|G_i|} \ti \psi (y)
\end{equation*}
and
\begin{equation*}
 \big(\phi_G \circ \rho\big)(y)= \prod_{i=1}^m (y_i^*)^{|G_i|} \ti \phi (y) 
\end{equation*}
with $\ti \psi$ and $\ti \phi$ regular functions on $\mb R_+^{N-1}$.
\end{lem}

\begin{proof}
 Both statements follow from \eqref{scaling} in Proposition \ref{condel}. As first step, we recall that \begin{equation*}
     \psi_G=\psi_{G_1} \psi_{G / G_1} + R_1                                                                                                                                         \end{equation*}
with $\psi_{G_1}$ and $\psi_{G / G_1}  $ depending only on $(y_1,\ldots,y_m)$ and $y_0$, respectively, and $R_1$ of degree $d_1>$ deg$(\psi_{G_1})=|G_1|$. Thus,
\begin{equation*}
 \psi_G\circ \rho=(y_1^*)^{|G_1|} \big( \psi_{G_1}' \psi_{G / G_1}' + R_1' \big)\rest{y_i^*=1} =: (y_1^*)^{|G_1|} \ti \psi_1    
\end{equation*}
where the prime $'$ denotes evaluation with the edge variables $(y_2,\ldots,y_m)$ still scaled by the map $\rho$.
Because $\psi_{G_1}$ further factorizes, 
\begin{equation*}
\psi_{G_1} '= \psi_{G_2} ' \psi_{G_1/G_2} ' + R_2',
\end{equation*}
and because $R_1$ is of degree $d_1>|G_1|>|G_2|$ also in $(y_2,\ldots,y_m)$, we can repeat the above argument to conclude
\begin{equation*}
 \big( \psi_{G_1} ' \psi_{G / G_1}' + R_1 ' \big)\rest{y_1^*=1} =   \big( ( \psi_{G_2} ' \psi_{G_1/G_2} ' + R_2') \psi_{G / G_1}' +R_1'\big)\rest{y_1^*=1}   =  (y_2^*)^{|G_2|}   \ti \psi_2.
\end{equation*}

After $m$ steps we arrive at the desired equation. The case $\phi_G$ works analogous.
\end{proof}

\begin{prop}\label{poles}
Consider a new face $\tau \subset \ti \sigma_G$, $\tau \sim (G,\m F)$ with $\m F=(G_1,\ldots,G_m)$. Then the differential form $\beta^*\omega_G$ has poles along $\tau$, one for each divergent $G_i$, of order $\frac{s_{G_i}}{2}+1$. Its regular part $\ti f_G$ satisfies
\begin{equation*}
 {{\ti f}_{G}}{\rest{\tau}}= \big( \psi_{G_m} \psi_{G_{m}/G_{m-1}} \cdots \psi_{G_2 / G_1} \big)^{-\frac{d}{2}} f_{G/G_1}  = \big( \prod_{i=1}^m \psi_{\gamma_i}\big)^{-\frac{d}{2}} f_{\gamma_0}.
\end{equation*}
\end{prop}

\begin{proof}
 Combining the result of Lemma \ref{rhosca} with the definition of $\omega_G$ in \eqref{omega} we find
 \begin{equation*}
  \big(f_G\circ \rho\big)(y)= \prod_{i=1}^m (y_i^*)^{-|G_i|\frac{d}{2}} \ti f_G
 \end{equation*}
with $\ti f_G$ regular. For the differential form $\nu_G$ we have
\begin{align*}
 \rho^* \nu_G & = \rho^*\big(\sum_{i=1}^N (-1)^i x_i dx_1 \wedge \ldots \wedge \widehat{dx_i} \wedge \ldots \wedge dx_N\big) \\
 & = \prod_{i=1}^m (y_i^*)^{|E(G_i)|-1}\nu_G.
\end{align*}
Putting everything together we conclude
\begin{equation}\label{locform}
 \rho^* \omega_G = \prod_{i=1}^m (y_i^*)^{-|G_i|\frac{d}{2}+|E(G_i)|-1} \ti f_G \nu_G= \prod_{i=1}^m (y_i^*)^{-\frac{s_{G_i}}{2}-1} \ti f_G \nu_G.
 \end{equation}
 
 Recall the notation from the proof of Lemma \ref{rhosca}. At $\tau$, where all $y_i^*=0$, every remainder term $R_i'$ in the factorizations of $\psi$ and $\phi$ vanishes. Therefore, the regular part $\ti f_G$ is given at $\tau$ by
 \begin{align*}
  {{\ti f}_{G}}{\rest{\tau}}= & \Big(  (\ti \psi_m)^{-\frac{d}{2}} \big( \frac{\ti \psi_m}{\ti \Xi_m} \big)^{-\frac{s_G}{2}}   \Big)\rest{\{y_i^*=0\}} \\
  = & \big( \psi_{G_m} \psi_{G_{m}/G_{m-1}} \cdots \psi_{G_2 / G_1} \psi_{G/G_1} \big)^{-\frac{d}{2}} \big( \frac{\ti \psi_m}{\ti \Xi_m} \big)^{-\frac{s_G}{2}}\rest{\{y_i^*=0\}} \\
  = &  \big( \prod_{i=1}^m \psi_{\gamma_i}\big)^{-\frac{d}{2}} \psi_{G/G_1}^{-\frac{d}{2}} \Big( \frac{\psi_{G/G_1}}{\Xi_{G/G_1}} \Big)^{-\frac{s_G}{2}} = \big( \prod_{i=1}^m \psi_{\gamma_i}\big)^{-\frac{d}{2}} f_{\gamma_0} \text{ as in \eqref{flagsforests}}
 \end{align*}
with 
\begin{equation*}
 \ti \Xi_m = \ti \phi_m +  (m_{c(e)}^2 + \sum_{e'\neq e} m_{c(e')}^2 \rho_{e'} )\ti \psi_m
\end{equation*}
since we are working in affine coordinates with $x_e=1$.
\end{proof}

\begin{rem}
 These poles are only superficial; the integrand $f_G$ might actually be better behaved. For example, consider the ''sunrise'' graph on two vertices connected by three internal edges with all $m_e=0$. For $d=4$ it has three divergent subgraphs, all satisfying $|\gamma|=1$ and $s_\gamma=0$. In this case, the second graph polynomial $\phi_G={\Xi_G}_{|{m=0}}$ scales as $x^{|\gamma|+1}$, so that there are no poles at the faces $\tau \sim (G,\gamma)$. In any case, we do not need to worry about this as incorporating trivial subtractions will not affect the final renormalisation.
 \end{rem}

  \begin{figure}[!htb]\label{sunrise}
  \begin{tikzpicture}
   \coordinate (v0) at (0,0);   
   \coordinate  (v1) at (2.3,0); 
   \draw (v0) -- (v1) node [xshift=-1.15cm,yshift=.8cm] {$1$};
   \draw (v0) to[out=66,in=114] (v1) node [xshift=-1.15cm,yshift=.18cm] {$2$};
   \draw (v0) to[out=-66,in=-114] (v1) node [xshift=-1.15cm,yshift=-0.4cm] {$3$};  
   \fill[black] (v0) circle (.066cm);
   \fill[black] (v1) circle (.066cm);
   \coordinate (p1) at (-0.66,0);
   \coordinate (p2) at (2.94,0);
   \draw (v0) -- (p1) node [xshift=.2cm,yshift=-0.2cm] {$p_1$};
   \draw (v1) -- (p2) node [xshift=-0.2cm,yshift=-0.2cm] {$p_2$};

   \coordinate  (d1) at (5,-0.3); 
   \coordinate (d2) at (5.5,-.8);   
   \coordinate  (d3) at (7.5,-.8);
   \coordinate (d4) at (8,-0.3);   
   \coordinate  (d5) at (6.9,1);
   \coordinate (d6) at (6.1,1); 
   
    \fill[fill=lightgray,opacity=0.6] (d1) -- (d2) -- (d3) -- (d4) -- (d5) -- (d6) -- (d1);
    
    \draw[red] (d1) -- (d2);
   \draw (d2) to (d3) node [xshift=-2.4cm,yshift=-.05cm] {$\m E_{1,2}$};
   \draw[red] (d3) -- (d4);
   \draw (d4) to (d5) node [xshift=1.1cm,yshift=-1.85cm] {$\m E_{2,3}$};
   \draw[red] (d5) -- (d6);
   \draw (d6) to (d1)node [xshift=1.5cm,yshift=1.5cm] {$\m E_{1,3}$};  
   \fill[red] (d1) circle (.04cm);
   \fill[red] (d2) circle (.04cm);
    \fill[red] (d3) circle (.04cm);
   \fill[red] (d4) circle (.04cm);
     \fill[red] (d5) circle (.04cm);
   \fill[red] (d6) circle (.04cm);
  \end{tikzpicture}  
\caption{The ''sunrise'' graph and its compactified cell $\ti \sigma_G \subset \ti X_{2,2}$}
 \end{figure}
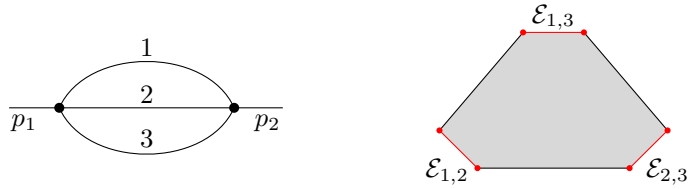
 
Proposition \ref{poles} shows that the new faces of $\ti \sigma_G$ encode all the divergent behaviour of $\tilde t_G$. Viewed as a meromorphic function of $d \in \mb C$ it has poles along the new faces $\tau$ of $\ti \sigma_G$ corresponding to divergent subgraphs, i.e.\ $\tau \sim (G=G_0 \supset G_1 \supset \ldots \supset G_m)$ with at least one $G_i$ divergent. In other words, $\tilde t_G$ diverges when applied to a test function whose support intersects a new face $\tau \subset \ti \sigma_G$ that is indexed by a divergent subgraph of $G$. Moreover, we can restrict to the case where $\tau \sim (G,\m F)$ and $\m F=(G_1 \supset G_2 \supset \ldots \supset G_m)$ with all $G_i$ divergent, other types of flags do not contain additional information about subdivergences.
\newline

In principle, we can now simply apply our favourite renormalisation scheme\footnote{See, for instance, Theorem $8$ in \cite{bk-as}; the formulae are rather long and not very enlightening per se and therefore omitted here.} to define a finite version $\tilde t^{\mathit{ren}}_G$ of $\tilde t_G$ on each $\ti \sigma_G$. Pushing this collection of distributions forward along the map $\beta$ produces then a piecewise distribution $ t^{\mathit{ren}}=\{t^{\mathit{ren}}_G\}$ on $X$, given by 
\begin{equation*}
 t_G^{\mathit{ren}}: C_c^{\infty}(\sigma_G) \ni \varphi \longmapsto \langle \beta_* \tilde t_G^{\mathit{ren}} \mid \varphi \rangle := \langle \tilde t_G^{\mathit{ren}} \mid \beta^*\varphi \rangle.
\end{equation*}
Notice, to obtain the value of $t^{\mathit{ren}}$ at the constant function $1 \equiv \chi \in C^{\infty}(X)$ we can circumvent this pushforward operation and thereby the need to approximate $\chi$ by a sequence of test functions. Instead, we simply evaluate each $\tilde t^{\mathit{ren}}_G$ at $1 \equiv \beta^*\chi \in C_c^{\infty}(\ti \sigma_G)$. Hence, the renormalised Feynman amplitude $\m A_n^{\mathit{ren}}$ is given by
\begin{equation*}
 \m A_n^{\mathit{ren}}: p \longmapsto  \sum_{G \in \m X} \langle \tilde t^{\mathit{ren}}_G(p) \mid 1 \rangle.
\end{equation*}
\newline

Another way to find $\tilde t^{\mathit{ren}}$, with the advantage of working entirely in the realm of piecewise distributions, is to treat renormalisation as an extension problem for distributions, as in the Epstein-Glaser method in the position space formulation of quantum field theory \cite{eg}. The key identity to start with is the local formula \eqref{locform} for $\tilde t_G$, valid in the vicinity of a new face $\tau \sim (G,\m F)$ described by the coordinate chart $\rho$ in \eqref{locobu},
\begin{equation*}
  \beta^* \omega_G \overset{loc.}= \rho^* \omega_G =  \prod_{ G_i \in \m F} (y_i^*)^{-\frac{s_{G_i}}{2}-1} \ti f_G \nu_G.
\end{equation*}
Let $U$ denote the chart domain of $\rho$. Equation \eqref{locform} allows to define a renormalisation operator $\m R_U$ that kills all poles in the local expression for $\beta^* \omega$ in $U$. Putting these pieces together using a partition of unity on $\ti \sigma_G$, subordinate to the charts $(U,\rho)$, produces then the desired renormalised distribution $\tilde t^{\mathit{ren}}_G$. 

In the following we sketch the construction of $\m R$ for graphs with only logarithmic subdivergences; details and proofs, including a discussion of a physicality condition on the proposed solution, can be found in \cite{mb1}. In the logarithmic case all the poles in \eqref{locform} are of first order only and the corresponding residua are all independent of the external data $(p,m)$. Thus, simple subtractions of these residua, one for each diverging coordinate direction, suffice to produce a finite expression for which the limit $d \rightarrow d_0$, i.e.\ $s_{G_i} \rightarrow 0$, exists. The operator $\m R_U$ is defined by
\begin{align*}
 \m R_U [ \rho^* \omega_G ] = & \m R_U [ \prod_{ G_i \in \m F} (y_i^*)^{-\frac{s_{G_i}}{2}-1} \ti f_G \nu_G ] \\
 :=&  \nu_G \prod_{ G_i \in \m F}(y_i^*)^{-1} \left(1 - \delta(y_i^*) \right) \ti f_G(p,m) \\
 = & \nu_G \prod_{G_i \in \m F} (y_i^*)^{-1} \sum_{\m H \subset \m F} (-1)^{|\m H|}  \prod_{G_i \in \m H}  \delta(y_i^*)\ti f_G(p,m),
\end{align*}
with $\delta$ the Dirac distribution, operating only on the regular part of $\rho^*\omega_G$ and the test function it is integrated with. 

The operator $\m R$ transforms $\tilde t_G$ into a well-defined distribution by removing all of its poles; if the support of a test function $\varphi$ intersects a divergent new face, then the corresponding pole of $\langle \tilde t_G \mid \varphi \rangle$ gets subtracted to make the integral finite. On the other hand, for every test function $\varphi$ with its support disjoint from any new divergent face we have $\langle \tilde t^{\mathit{ren}}_G \mid \varphi \rangle = \langle \tilde t_G \mid \varphi \rangle $ because all the subtracted terms vanish, $\langle \delta \mid \varphi \rangle = \varphi_{|{y_i^*=0}}=0$. Thus, $\tilde t^{\mathit{ren}}_G$ defines an extension of $\tilde t_G$ from the complement of all divergent new faces of $\ti \sigma_G$ to the whole cell.

\begin{example}
 The formula for $\m R$ is best understood by studying a concrete example. For the Dunce's cap graph we computed in Section \ref{feyint} 
 \begin{equation*}
  \psi_G(x_1, \ldots, x_4)=x_3x_4 + x_2x_3 + x_2x_4 + x_1x_3 + x_1x_4. 
 \end{equation*}
Locally in an affine chart where $x_1=1$ we have $\mathbf{y}_0=y_0=x_2$, $\mathbf{y}_1=(y_1^*,y_1)=(x_3,x_4)$ and
 \begin{equation*}
  \langle \rho^* \omega_G \mid \varphi \rangle =\int_{\mb R_+^3} dy (y_1^*)^{-1} \frac{\varphi(y)}{(y_1y_1^* + y_0 + y_0 y_1 + 1 + y_1)^{2}}.
 \end{equation*}
Therefore, $\langle \m R_U[\rho^* \omega_G] \mid \varphi \rangle$ is given by
\begin{equation*}
 \int_{\mb R_+^3} dy (y_1^*)^{-1} \Big( \frac{\varphi(y)}{(y_1y_1^* + y_0 + y_0 y_1 + 1 + y_1)^{2}} - \frac{\varphi(y)_{|{y_1^*=0}}}{( y_0 + y_0 y_1 + 1 + y_1)^{2}} \Big),
\end{equation*}
 which is a finite expression.
\end{example}

In the general case $\beta^*\omega_G$ may have poles of arbitrary high order. These can be reduced to poles of first order by partial integrations at the cost of boundary terms which in turn are then cured by subtracting terms of the Taylor expansion around a renormalisation point. Then simple subtractions as above allow to define a finite renormalised distribution on $\sigma_G$. Again, for details the reader is refer to the exhaustive exposition in \cite{bk-as}.

\bibliography{ref}
\bibliographystyle{alpha}
\end{document}